\documentclass[conference]{IEEEtran}

\usepackage{cite}
\usepackage{graphicx}
\usepackage{amsmath}
\usepackage{amsfonts}
\usepackage{array}
\usepackage{amssymb}
\usepackage{subfigure}
\newtheorem{lemma}{\textbf{Lemma}}
\newtheorem{theorem}{\textbf{Theorem}}
\newtheorem{corollary}{\textbf{Corollary}}
\newtheorem{remark}{Remark}

\begin{document}

\title{Achievable Throughput of Multi-mode Multiuser MIMO with Imperfect CSI Constraints}
\author{\authorblockN{Jun~Zhang, Marios Kountouris, Jeffrey~G.~Andrews and Robert~W.~Heath~Jr.}
\authorblockA{Wireless Networking and Communications Group\\
Department of Electrical and Computer Engineering\\
The University of Texas at Austin, Austin, TX 78712-0240\\
Email: \{jzhang2, mkountouris, jandrews, rheath\}@ece.utexas.edu}}

\maketitle

\begin{abstract}
For the multiple-input multiple-output (MIMO) broadcast channel with
imperfect channel state information (CSI), neither the capacity nor
the optimal transmission technique have been fully discovered. In
this paper, we derive achievable ergodic rates for a MIMO fading
broadcast channel when CSI is delayed and quantized. It is shown
that we should not support too many users with spatial division
multiplexing due to the residual inter-user interference caused by
imperfect CSI. Based on the derived achievable rates, we propose a
multi-mode transmission strategy to maximize the throughput, which
adaptively adjusts the number of active users based on the channel
statistics information.
\end{abstract}

\section{Introduction}
For the multiple-input multiple-output (MIMO) broadcast channel,
channel state information at the transmitter (CSIT) is required to
separate the spatial channels for different users and achieve the
full spatial multiplexing gain. CSIT, however, is difficult to get
and is never perfect. Neither the capacity nor the optimal
transmission technique have been fully discovered. Linear precoding
combined with limited feedback \cite{LovHea08JSAC} is a practical
option, which has drawn lots of interest recently
\cite{Jin06IT,CaiJin07Submit,DinLov07Tsp,YooJin07JSAC}. The main
finding is that the full spatial multiplexing gain can be obtained
with carefully designed feedback strategy and sufficient feedback
bits that grow linearly with signal-to-noise ratio (SNR) (in dB) and
the number of transmit antennas.

In most systems, the number of feedback bits per user is fixed. In
addition, there are other CSIT imperfections, such as estimation
error and feedback delay. All of these make the system throughput
limited by the residual inter-user interference at high SNR
\cite{LapSha05Allerton}. A simple approach to solve this problem is
to adaptively switch between the single-user (SU) and multi-user
(MU) modes, as the SU mode does not suffer from the residual
interference at high SNR. SU/MU mode switching algorithms for the
random beamforming system were proposed in
\cite{YeuPar07Accept,KouGes07SPAWC}, where each user feeds back its
preferred mode and the channel quality information. Mode switching
for systems with zero-forcing (ZF) precoding and limited feedback
was investigated in \cite{KouFra07ICASSP,TriBoc07VTC}, where the
switching is performed during the scheduling process with properly
designed channel quality information feedback.

The above mentioned SU/MU mode switching algorithms are based on
instantaneous CSIT, and require feedback from each user in each time
slot. In \cite{ZhaAnd08Allerton,ZhaAnd}, a SU/MU mode switching
algorithm was proposed for the system with delayed and quantized
CSIT. The mode switching is based on the statistics of the channel
information, including the average SNR, the normalized Doppler
frequency, and the codebook size, which are easily available at the
transmitter. But it only switches between the SU mode and the full
MU mode that serves the maximum number of users that can be
supported, i.e. it is a dual-mode switching strategy.

In this paper, we consider a MIMO broadcast channel with delayed and
quantized CSIT, with the amount of delay and the size of the
quantization codebook fixed. We derive an achievable ergodic rate
for each \emph{transmission mode}, denoting the number of users
served by spatial division multiplexing. It is shown that the number
of active users is related to the transmit array gain, spatial
multiplexing gain, and the residual inter-user interference. The
full MU mode normally should not be activated, as it suffers from
the highest interference while provides no array gain. A multi-mode
transmission strategy is proposed to adaptively select the active
mode to maximize the throughput.


\section{System Model}\label{Sec:SysMod}
We consider a MIMO broadcast channel with $N_t$ antennas at the
transmitter and $U$ single-antenna receivers. Each time slot, the
transmitter determines the number of users to be served, denoted as
\emph{the transmission mode M}, $1\leq{M}\leq{N_t}$.
Eigen-beamforming is applied for the SU mode ($M=1$), which is
optimal with perfect CSIT. ZF precoding is used for the MU mode
($1<M\leq{N_t}$), as it is possible to derive closed-form results
due to its simple structure, and it is optimal among the set of all
linear precoders at asymptotically high SNR \cite{Jin05ISIT}. The
discrete-time complex baseband received signal at the $u$-th user in
mode $M$ at time $n$ is given as
\begin{equation}
{y}_u[n]=\mathbf{h}^*_u[n]\sum_{u'=1}^{M}\mathbf{f}_{u'}[n]{x}_{u'}[n]+{z}_u[n],
\end{equation}
where $\mathbf{h}_u[n]$ is the channel vector for the $u$-th user,
and ${z}_u[n]$ is the normalized complex additive Gaussian noise,
$z_u[n]\sim\mathcal{CN}(0,1)$. ${x}_u[n]$ and $\mathbf{f}_u[n]$ are
the transmit signal and precoding vector for the $u$-th user. The
transmit power constraint is
$\mathbb{E}\left[\sum_{u=1}^M|x_u[n]|^2\right]=P$, and we assume
equal power allocation among different users. As the noise is
normalized, $P$ is also the average SNR.

To assist the analysis, we assume that the channel $\mathbf{h}_u[n]$
is well modeled as a spatially white Gaussian channel, with entries
${h}_{i}[n]\sim\mathcal{CN}(0,1)$. We assume perfect CSI at the
receiver and the transmitter obtains CSI through limited feedback.
In addition, there is delay in the available CSIT. The models for
delay and limited feedback are presented as follows.

\subsection{CSI Delay Model}
We consider a stationary ergodic Gauss-Markov block fading regular
process (or auto regressive model of order $1$), where the channel
stays constant for a symbol duration and changes from symbol to
symbol according to
\begin{equation}\label{ChDelay}
\mathbf{h}[n]=\rho{\mathbf{h}}[n-1]+\mathbf{e}[n],
\end{equation}
where $\mathbf{e}[n]$ is the channel error vector, with i.i.d.
entries $e_{i}[n]\sim\mathcal{CN}(0,\epsilon_e^2)$, and it is
uncorrelated with $\mathbf{h}[n-1]$ and i.i.d. in time. We assume
the CSI delay is of one symbol. For the numerical analysis, the
classical Clarke's isotropic scattering model will be used as an
example, for which the correlation coefficient is
$\rho=J_0(2\pi{f_d}T_s)$ with Doppler spread $f_d$ \cite{Cla68},
where $T_s$ is the symbol duration and $J_0(\cdot)$ is the zero-th
order Bessel function of the first kind. The variance of the error
vector is $\epsilon_e^2=1-\rho^2$. The value $f_dT_s$ is the
normalized Doppler frequency. Note that this model can also be used
to model the estimation or prediction error, which makes our results
more general.

\subsection{Channel Quantization}
The channel direction information is fed back using a quantization
codebook known at both the transmitter and receiver. The
quantization is chosen from a codebook of unit norm vectors of size
$L=2^B$. Each user uses a different codebook to avoid the same
quantization vector. The codebook for user $u$ is
$\mathcal{C}_u=\{\mathbf{c}_{u,1},\mathbf{c}_{u,2},\cdots,\mathbf{c}_{u,L}\}$.
Each user quantizes its channel direction to the closest codeword,
and the closeness is measured by the inner product. Therefore, the
index of channel for user $u$ is
\begin{equation}
I_u=\arg\max_{1\leq{\ell}\leq{L}}|\tilde{\mathbf{h}}_u^*\mathbf{c}_{u,\ell}|,
\end{equation}
where $\tilde{\mathbf{h}}_u=\frac{\mathbf{h}_u}{\|\mathbf{h}_u\|}$
is the channel direction. Random vector quantization (RVQ) is used
to facilitate the analysis, where each quantization vector is
independently chosen from the isotropic distribution on the
$N_t$-dimensional unit sphere. The codebook based on RVQ is
asymptotically optimal in probability as $N_t$,
$B\rightarrow\infty$, with
$\frac{B}{N_t}\rightarrow\hat{r}\in\mathbb{R}^+$ \cite{DaiLiu08IT}.

\subsection{Multi-mode Transmission}
In this paper, we consider a homogeneous network where all the users
have the same average SNR, mobility, delay and feedback bits. The
transmitter will determine how many users to serve, i.e. the active
mode $M^\star$, and then $M^\star$ users are selected from the total
$U$ users randomly or with round-robin scheduling. The mode
selection is based on the information including average SNR,
normalized Doppler frequency, and the quantization codebook size.
Only the selected $M^\star$ users need to feed back their
instantaneous channel information. Such user scheduling is suitable
for the system with scheduling independent of the channel status,
such as round-robin or the one based on the queue length, or the
small system with user number roughly equal to the transmit
antennas, such as in a cooperative communication network.

\section{Throughput Analysis and Mode Selection}\label{Sec:MultiMode}
In this section, we derive the achievable ergodic rate for each
mode, based on which the active mode can be selected. Both perfect
and imperfect CSIT systems are investigated.

\subsection{Perfect CSIT}

\subsubsection{SU Mode (Eigen-beamforming), $M=1$}
With perfect CSIT, the beamforming (BF) vector is the channel
direction vector, $\mathbf{f}[n]=\tilde{\mathbf{h}}[n]$. The ergodic
achievable rate is
\begin{align}\label{C_BF}
R_{CSIT}(1)&=R_{BF}(\gamma,N_t)\notag\\
&=\mathbb{E}_{\mathbf{h}}\left[\log_2\left(1+\gamma|\mathbf{h}^*[n]\mathbf{f}[n]|^2\right)\right]\notag\\
&=\log_2(e)e^{1/\gamma}\sum_{k=0}^{N_t-1}\frac{\Gamma(-k,1/\gamma)}{\gamma^k},
\end{align}
where $\Gamma(\alpha,x)=\int_x^\infty{t^{\alpha-1}e^{-t}dt}$ is the
complementary incomplete gamma function. The BF system provides a
transmit array gain $N_t$ as
$\mathbb{E}_{\mathbf{h}}\left[\gamma|\mathbf{h}^*[n]\mathbf{f}[n]|^2\right]=N_t\gamma$.

\subsubsection{MU Mode (Zero-forcing), $1<M\leq{N_t}$}
The received SINR for the $u$-th user in a linear precoding MU-MIMO
system of mode $M$ is given by
\begin{align}
\mbox{SINR}_{u}(M)=\frac{\frac{\gamma}{M}|\mathbf{h}^*_u[n]\mathbf{f}_u[n]|^2}{1+\frac{\gamma}{M}\sum_{u'\neq{u}}|\mathbf{h}_u^*[n]\mathbf{f}_{u'}[n]|^2}.
\end{align}

Denote
${\mathbf{H}}[n]=[{\mathbf{h}}_1[n],{\mathbf{h}}_2[n],\cdots,{\mathbf{h}}_U[n]]^*$,
and the pseudo-inverse of ${\mathbf{H}}[n]$ as
$\mathbf{F}[n]=\mathbf{H}^\dag[n]=\mathbf{H}^*[n](\mathbf{H}[n]\mathbf{H}^*[n])^{-1}$.
The precoding vector for the $u$-th user is obtained by normalizing
the $u$-th column of $\mathbf{F}[n]$. Therefore,
$\mathbf{h}^*_u[n]\mathbf{f}_{u'}[n]=0,\,\forall{u\neq{u'}}$, i.e.
there is no inter-user interference and each user gets an equivalent
interference-free channel. The SINR for the $u$-th user becomes
\begin{align}\label{ZFSINR}
\mbox{SINR}_{ZF,u}(M)={\frac{\gamma}{M}|\mathbf{h}^*_u[n]\mathbf{f}_u[n]|^2}.
\end{align}
Due to the isotropic nature of i.i.d. Rayleigh fading, such
orthogonality constraints to precancel inter-user interference
consume $M-1$ degrees of freedom at the transmitter. As a result,
the effective channel gain of each parallel channel is a chi-square
random variable with $2(N_t-M+1)$ degrees of freedom
\cite{LeeJin07IT}, i.e.
$|\mathbf{h}_u^*[n]\mathbf{f}_u[n]|^2\sim\chi^2_{2(N_t-M+1)}$.
Therefore, the channel for each user is equivalent to a diversity
channel with order $N_t-M+1$ and effective SNR
$\gamma_u=\frac{\gamma}{M}$. The achievable rate for the $u$-th user
in mode $M$ is
\begin{align}\label{R_ZFu}
R_{ZF,u}(\gamma_u,M)&=\mathbb{E}_{\mathbf{h}}\left[\log_2\left(1+\frac{\gamma}{M}|\mathbf{h}_u^*[n]\mathbf{f}_u[n]|^2\right)\right]\notag\\
&=R_{BF}(\gamma_u,N_t-M+1),
\end{align}
where $R_{BF}(\cdot,\cdot)$ is given in \eqref{C_BF}. The achievable
sum rate for the ZF system of mode $M$ is
\begin{align}\label{R_ZF}
R_{CSIT}(M)&=\sum_{u=1}^MR_{ZF,u}(\gamma_u,M)\notag\\
&{=}MR_{BF}(\gamma_u,N_t-M+1).
\end{align}
When $M=1$, this reduces to \eqref{C_BF}.

\subsubsection{Mode Selection}
From \eqref{R_ZF}, the system in mode $M$ provides a spatial
multiplexing gain of $M$ and an array gain of $N_t-M+1$ for each
user. As $M$ increases, the achievable spatial multiplexing gain
increases but the array gain decreases. Therefore, there is a
tradeoff between the achievable array gain and the spatial
multiplexing gain. From \eqref{R_ZF}, the mode that achieves a
higher throughput for the given average SNR can be determined as
\begin{equation}\label{CSIT_MS}
M^\star=\arg\max_{1\leq{M}\leq{N_t}}R_{CSIT}(M).
\end{equation}

\subsection{Imperfect CSIT}
In this section, we consider a system with both delay and channel
quantization. As it is difficult to derive the exact achievable rate
for such a system, we provide accurate closed-form approximations
for mode selection.

\subsubsection{SU Mode (Eigen-Beamforming)}
With delay and channel quantization, the BF vector is based on the
delayed version of the quantized channel direction,
$\mathbf{f}^{(QD)}=\hat{\mathbf{h}}[n-1]$.

To get a good approximation for the achievable rate for the SU mode,
we first make the following approximation on the instantaneous
received SNR
\begin{align}
\mbox{SNR}_{BF}^{(QD)}&=P|\mathbf{h}^*[n]\hat{\mathbf{h}}[n-1]|^2\notag\\
&=P|(\rho\mathbf{h}[n-1]+\mathbf{e}[n])^*\hat{\mathbf{h}}[n-1]|^2\notag\\
&\approx{P}\rho^2|\mathbf{h}^*[n-1]\hat{\mathbf{h}}[n-1]|^2,\label{SNR_BF}
\end{align}
i.e. we remove $\mathbf{e}[n]$, which is normally insignificant
compared to $\rho\mathbf{h}[n-1]$. Eq. \eqref{SNR_BF} is equivalent
for a limited feedback system. From \cite{AuLov07Twc}, the
achievable rate of the limited feedback BF system is given by
\begin{align}
&R_{BF}^{(Q)}(P)=\log_2{(e)}\left(e^{1/P}\sum_{k=0}^{N_t-1}E_{k+1}\left(\frac{1}{P}\right)\right.\notag\\
&\left.-\int_0^1\left(1-(1-x)^{N_t-1}\right)^{2^B}\frac{N_t}{x}e^{1/P{x}}E_{N_t+1}\left(\frac{1}{P{x}}\right)dx\right),
\end{align}
where $E_n(x)=\int_1^\infty{e}^{-xt}x^{-n}dt$ is the $n$-th order
exponential integral. So the achievable rate of the BF system with
both delay and channel quantization can be approximated as
\begin{equation}\label{Approx_BFQD}
R_{BF}^{(QD)}(P)\approx{R}_{BF}^{(Q)}(\rho^2P).
\end{equation}

\subsubsection{MU Mode (Zero-Forcing)}
For the MU mode with imperfect CSIT, the precoding vectors are
designed based on the quantized channel directions with delay, which
achieve
$\hat{\mathbf{h}}^*_u[n-1]\mathbf{f}^{(QD)}_{u'}[n]=0,\,\forall{u\neq{u'}}$.
The SINR for the $u$-th user in mode $M$ is given as
\begin{equation}\label{ZFSINR_QD}
\gamma_{ZF,u}^{(QD)}(M)=\frac{\frac{P}{M}|\mathbf{h}^*_u[n]\mathbf{f}^{(QD)}_u[n]|^2}{1+\frac{P}{M}\sum_{u'\neq{u}}|\mathbf{h}_u^*[n]\mathbf{f}^{(QD)}_{u'}[n]|^2}.
\end{equation}
Following \emph{Theorem 3} in \cite{ZhaAnd08Allerton}, the average
residual inter-user interference for the $u$-th user is
\begin{align}\label{I_QD}
&\mathbb{E}_{\mathbf{h}}\left[\frac{P}{M}\sum_{u'\neq{u}}|\mathbf{h}_u^*[n]\mathbf{f}^{(QD)}_{u'}[n]|^2\right]\notag\\
=&\left(1-\frac{1}{M}\right)P\left(\rho_u^2\frac{N_t}{N_t-1}2^{-\frac{B}{N_t-1}}+\epsilon_{e,u}^2\right).
\end{align}
\begin{remark}
The residual interference depends on delay, codebook size, $N_t$,
and $M$. It increases with delay, and decreases with the codebook
size. At high SNR, it makes the system interference-limited. With
other parameters fixed, residual interference increases as $M$
increases, and mode selection will take this into consideration.
\end{remark}

To approximate the achievable ergodic rate, we first approximate the
signal term as
\begin{align}\label{s_QD}
P_S=&\frac{P}{M}|\mathbf{h}^*_u[n]\mathbf{f}^{(QD)}_u[n]|^2\notag\\
=&\frac{P}{M}|(\rho_u\mathbf{h}_u[n-1]+\mathbf{e}_u[n])^*\mathbf{f}^{(QD)}_u[n]|^2\notag\\
\stackrel{(a)}{\approx}&\frac{P}{M}|\rho_u\mathbf{h}^*_u[n-1]\mathbf{f}^{(QD)}_u[n]|^2\notag\\
\stackrel{(b)}{\approx}&\frac{P}{M}|\rho_u\|\mathbf{h}_u[n-1]\|\hat{\mathbf{h}}^*_u[n-1]\mathbf{f}^{(QD)}_u[n]|^2,
\end{align}
where step (a) removes $\mathbf{e}_u^*[n]\mathbf{f}^{(QD)}_u[n]$,
which is normally very small. Step (b) approximates the actual
channel direction by the quantized version, which is justified for
small quantization error. As $\hat{\mathbf{h}}_u[n-1]$ is
independent from other users, similar to the system with perfect
CSIT,
$\|\mathbf{h}_u[n-1]\|\cdot|\hat{\mathbf{h}}^*_u[n-1]\mathbf{f}^{(QD)}[n]|^2\sim\chi^2_{2(N_t-M+1)}$.

\begin{figure*}[!t]
\normalsize
\begin{align}\label{SINR_ApproxDQ}
\gamma_{ZF,u}^{(QD)}\approx\frac{\frac{P}{M}\rho^2\|\mathbf{h}_u[n-1]\|\cdot|\hat{\mathbf{h}}^*_u[n-1]\mathbf{f}^{(QD)}[n]|^2}{1+\frac{P}{M}\sum_{u'\neq{u}}\rho^2|\mathbf{h}_u^*[n-1]\mathbf{f}^{(QD)}_{u'}[n]|^2+\frac{P}{M}\sum_{u'\neq{u}}|\mathbf{e}_u^*[n]\mathbf{f}_{u'}[n]|^2}
\end{align}
\begin{equation}\label{R_ZFDQ}
R_{QD,u}(M)\approx\log_2(e)\sum_{i=0}^{N_t-L-1}\sum_{j=1}^2\sum_{k=0}^{L-1}\sum_{l=0}^i\frac{a_k^{(j)}(l+k)!}{l!(i-l)!}\alpha^{l+k-i+1}\hat{I}\left(\frac{1}{\alpha},\frac{\alpha}{\delta_j},i,l+k+1\right)
\end{equation}
\begin{equation}\label{R_ZFD}
R_{D,u}(M)\approx\log_2(e)\sum_{i=0}^{N_t-L-1}\sum_{l=0}^i{{L+l-1}\choose{l}}\frac{\alpha^{L+l-i}}{\beta^L(i-l)!}\cdot{\hat{I}}\left(\frac{1}{\alpha},\frac{\alpha}{\beta},i,L+l\right)
\end{equation}
\hrulefill \vspace*{4pt}
\end{figure*}

The received SINR for the $u$-th user can then be approximated as in
\eqref{SINR_ApproxDQ}, on the top of the next page. The
approximation for the denominator comes from removing the terms with
both $\mathbf{e}_u[n]$ and $\mathbf{f}^{(QD)}_{u'}[n]$. For the
interference terms in \eqref{SINR_ApproxDQ}, as
$\mathbf{e}_u[n]\sim\mathcal{CN}(\mathbf{0},\epsilon_{e,u}^2\mathbf{I})$,
$|\mathbf{f}_{u'}[n]|^2=1$, and they are independent,
$|\mathbf{e}_u^*[n]\mathbf{f}_{u'}[n]|^2$ is an exponential random
variable with mean $\epsilon_e^2$. The following lemma in
\cite{ZhaAnd} provides the distribution of the term
$|\mathbf{h}_u^*[n-1]\mathbf{f}^{(QD)}_{u'}[n]|^2$.
\begin{lemma}\label{lemma_Iq}
Based on the quantization cell approximation \cite{YooJin07JSAC},
the interference term due to quantization,
$|\mathbf{h}_u^*[n-1]\mathbf{f}^{(QD)}_{u'}[n]|^2$, can be
approximated as an exponential random variable with mean
$\delta=2^{-\frac{B}{N_t-1}}$.
\end{lemma}
\renewcommand{\labelenumi}{(\roman{enumi})}
\begin{remark}
The residual interference terms due to both delay and quantization
are exponential random variables, which means the delay and
quantization error have equivalent effects but with different means.
\end{remark}

Based on the distributions of the interference terms, we can get the
following theorem on the achievable rate for the MU mode.
\begin{theorem}\label{thm_R_DQ}
The achievable ergodic rate for the $u$-th user in mode $M$ ($M>1$)
with both delay and channel quantization can be approximated by
\eqref{R_ZFDQ}, where $\alpha=\frac{\rho_u^2P}{M}$,
$\delta_1=\frac{\rho_u^2P\delta}{M}$,
$\delta_2=\frac{\epsilon_{e,u}^2P}{M}$, $L=M-1$, $\hat{I}$ is the
following integral
\begin{equation}\label{I3}
\hat{I}(a,b,m,n) =\int_0^\infty\frac{x^me^{-ax}}{(x+b)^n(x+1)}dx,
\end{equation}
and
\begin{align}
a^{(1)}_i&=\frac{\mathcal{A}}{\delta_1^{i+1}}\left(\frac{\delta_1}{\delta_1-\delta_2}\right)^L\left(\frac{\delta_2}{\delta_2-\delta_1}\right)^{L-1-i},\notag\\
a^{(2)}_i&=\frac{\mathcal{A}}{\delta_2^{i+1}}\left(\frac{\delta_2}{\delta_2-\delta_1}\right)^L\left(\frac{\delta_1}{\delta_1-\delta_2}\right)^{L-1-i}\notag,
\end{align}
with $\mathcal{A}=\frac{1}{(L-1)!}\frac{(2(L-1)-i)!}{i!(L-1-i)!}$.
\end{theorem}

\begin{proof}
This approximation is derived by assuming the interference terms are
pair-wise independent, and also independent of the signal term. Then
the CDF of the SINR can be derived, following which the achievable
rate can be calculated. A closed form expression can be derived for
$\hat{I}(a,b,m,n)$, which is not provided due to space limitation.
\end{proof}

As a special case, with only delay, following \eqref{s_QD} and
\eqref{SINR_ApproxDQ} the received SINR for the $u$-th user is
approximated as
\begin{align}\label{SINR_Dapprox}
\gamma_{ZF,u}^{(D)}\approx\frac{\frac{P}{M}|\rho_u\mathbf{h}^*_u[n-1]\mathbf{f}_u[n]|^2}{1+\frac{P}{M}\sum_{u'\neq{u}}|\mathbf{e}_u^*[n]\mathbf{f}_{u'}[n]|^2},
\end{align}
for which the achievable ergodic rate is provided as follows.
\begin{corollary}\label{thm_R_D}
The achievable ergodic rate for the $u$-th user in the delayed
system of mode $M$ ($M>1$) can be approximated by \eqref{R_ZFD},
where $\alpha=\frac{\rho_u^2P}{M}$,
$\beta=\frac{\epsilon_{e,u}^2P}{M}$, $L=M-1$, and $\hat{I}$ is the
integral in \eqref{I3}.
\end{corollary}

From \emph{Lemma \ref{lemma_Iq}}, the effects of delay and channel
quantization are equivalent, so the result in this corollary also
applies for the limited feedback system, replacing
$\epsilon_{e,u}^2$ by $\delta=2^{-\frac{B}{N_t-1}}$.

\subsection{Mode Selection}
Based on \eqref{Approx_BFQD} and \eqref{R_ZFDQ}, the active mode in
the system with both delay and channel quantization is selected
according to
\begin{equation}\label{eq_mode}
M^\star=\arg\max_{1\leq{M}\leq{N_t}}R_{QD}(M),
\end{equation}
where $R_{QD}(M)=\sum_{u=1}^MR_{QD,u}(M)$.
\begin{remark}
Considering \eqref{I_QD}, \eqref{s_QD}, and \eqref{R_ZFDQ}, the mode
$M$ is now related to the residual interference, the transmit array
gain, and the spatial multiplexing gain. The multi-mode transmission
is to balance between these effects to improve the system
throughput.
\end{remark}

\section{Numerical Results: Verification of Analysis and Key Observations}\label{Sec:Num}
\begin{figure}[htb]
\centering
\includegraphics[width=3.5in]{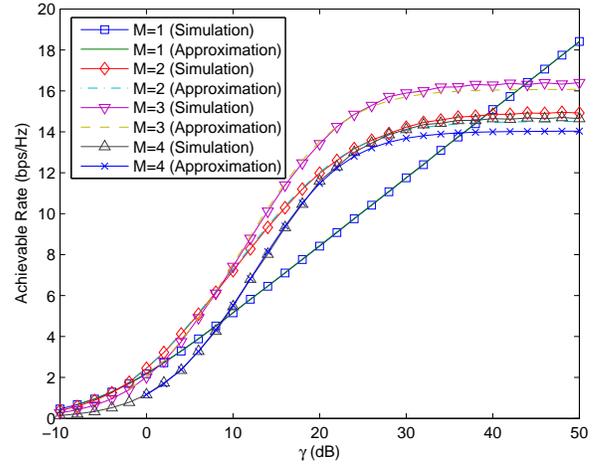}
\caption{Comparison of simulations and approximations for different
$M$, $N_t=4$, mobility $v=10$ km/hr, $T_s=1$ msec,
$B=18$.}\label{fig_SimvsCal_DQ}
\end{figure}

We first verify the derived approximations for the achievable rates
in Fig. \ref{fig_SimvsCal_DQ}. We see that the approximation is very
accurate at low to medium SNRs. At high SNR, the approximation
becomes a lower bound, and the accuracy decreases as $M$ increases.
Interestingly, we see that the mode $M=3$ always provides a higher
throughput than the full MU mode $M=4$. This is due to the fact that
the full mode has the highest level of residual interference while
provides no array gain. Therefore, it is desirable to do spatial
division multiplexing for fewer than $N_t$ users, which provides
addition array gain and reduces the residual interference for each
user. We see that the proposed multi-mode transmission provides a
significant throughput gain over the SU-MIMO system ($M=1$). In
addition, it is able to provide a throughput gain around $2$ bps/Hz
over the dual-mode switching \cite{ZhaAnd08Allerton,ZhaAnd} at
medium SNR.

\begin{figure*}[ht]
\centerline{\subfigure[Different $f_dT_s$,
$B=18$.]{\includegraphics[width=3.2in]{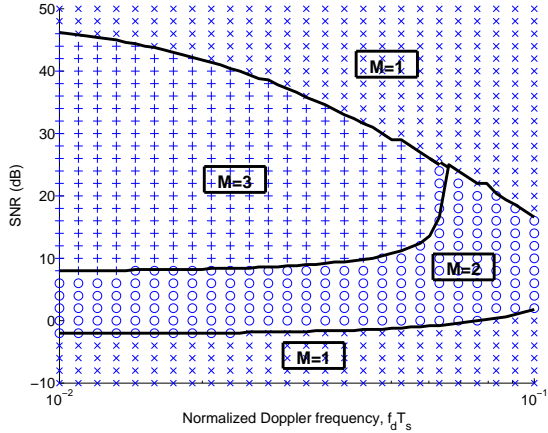}
\label{figSwitchRegion_D_B18}} \hfil \subfigure[Different $B$,
$f_c=2GHz$, $T_s=1$ msec., and $v=10$
km/hr.]{\includegraphics[width=3.2in]{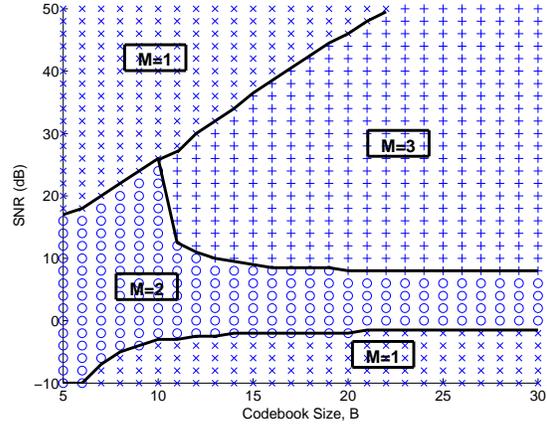}
\label{figSwitchRegion_Q_v10}}} \caption{Operating regions for
different modes with both CSI delay and channel quantization,
$U=N_t=4$. The mode $M=i$ means that there are $i$ active users. In
this plot, `$\times$' is for $M=1$, `$\circ$' is for $M=2$, `$+$' is
for $M=3$, and `$\square$' is for $M=4$.} \label{figSwitchRegion_DQ}
\end{figure*}

From \eqref{eq_mode}, we can determine the active mode $M^\star$ for
a given scenario. Accordingly, the operating regions for different
modes can be plotted for different system parameters. Fig.
\ref{figSwitchRegion_D_B18} and Fig. \ref{figSwitchRegion_Q_v10}
show the operating regions for different $f_dT_s$ and different $B$,
respectively. There are several key observations.
\begin{enumerate}
\item For the given $f_dT_s$ and $B$, the SU mode ($M=1$) will be active at both
low and high SNRs, due to its array gain and the robustness to
imperfect CSIT, respectively.

\item For higher MU modes to be active, $f_dT_s$ needs to be small
while $B$ needs to be large. Specifically, to get the mode $M=3$
activated, we need $f_dT_s<0.57$ with $B=18$ as in Fig.
\ref{figSwitchRegion_D_B18}, and need $B>10$ with $v=10$ km/hr,
$f_c=2$ GHz, and $T_s=1$ msec as in Fig.
\ref{figSwitchRegion_Q_v10}.

\item The highest mode $M=N_t$ is not active with the considered
parameters. We find that it is possible to be active only when
$f_dT_s\lesssim{0.01}$ and $B$ is large enough.
\end{enumerate}

\section{Conclusions}\label{Sec:Conclusion}
In this paper, we derive the achievable ergodic rates for the MIMO
broadcast channel with imperfect CSIT. Due to residual inter-user
interference, it is not desirable to serve too many users with
spatial division multiplexing. Based on the derived rates, a
multi-mode transmission strategy is proposed that adaptively adjusts
the number of active users. It is based on the channel statistics
information, and only the selected users need to feed back the
instantaneous channel information.

\bibliographystyle{IEEEtran}
\bibliography{CSIDelay}

\end{document}